\documentclass[%
 reprint,
superscriptaddress,
 amsmath,amssymb,
 aps,
pra,
citeautoscript,
floatfix
]{revtex4-2}

\usepackage{graphicx}
\usepackage{dcolumn}
\usepackage{bm}
\usepackage{outlines}
\usepackage{braket}
\usepackage{float}
\usepackage{makecell}
\usepackage[colorlinks,citecolor=blue]{hyperref}
\usepackage{mathtools}
\usepackage[dvipsnames]{xcolor}
\usepackage{url}
\usepackage[normalem]{ulem} 
\usepackage{siunitx}
\usepackage{array}
\usepackage{dsfont}
\usepackage{amsthm}
\usepackage{xpatch}
\usepackage{xcolor}

\definecolor{myred}{RGB}{236, 17, 0}
\definecolor{myorange}{RGB}{236, 137, 0}

\xpatchbibmacro{doi+eprint+url}{%
  \printfield{doi}%
}{%
  \iffieldundef{doi}{}{\href{https://doi.org/\StrSubstitute{\thefield{doi}}{ }{}}{\nolinkurl{\thefield{doi}}}}%
}{}{}

\newcommand{\ketbra}[1]{\ket{#1}\bra{#1}}
\newcommand{\id}{\mathds{1}}

\newtheorem{Proposition}{Proposition}

\begin{document}

\title{Complete Self-Testing of a System of Remote Superconducting Qubits}


\author{Simon Storz}
 \email{simon.storz@phys.ethz.ch}
 \affiliation{Department of Physics, ETH Zurich, 8093 Zurich, Switzerland}
 \affiliation{Quantum Center, ETH Zurich, 8093 Zurich, Switzerland}

\author{Anatoly Kulikov}
 \affiliation{Department of Physics, ETH Zurich, 8093 Zurich, Switzerland}
 \affiliation{Quantum Center, ETH Zurich, 8093 Zurich, Switzerland}
 
\author{Josua D. Schär}
 \affiliation{Department of Physics, ETH Zurich, 8093 Zurich, Switzerland}
 \affiliation{Quantum Center, ETH Zurich, 8093 Zurich, Switzerland}
 
\author{Victor Barizien}
 \affiliation{Institute of Theoretical Physics, University of Paris-Saclay, CEA, CNRS, Gif-sur-Yvette, France}

\author{Xavier Valcarce}
 \affiliation{Institute of Theoretical Physics, University of Paris-Saclay, CEA, CNRS, Gif-sur-Yvette, France}

\author{Florence Berterottière}
 \affiliation{Department of Physics, ETH Zurich, 8093 Zurich, Switzerland}
 \affiliation{Quantum Center, ETH Zurich, 8093 Zurich, Switzerland}

\author{Nicolas Sangouard}
 \affiliation{Institute of Theoretical Physics, University of Paris-Saclay, CEA, CNRS, Gif-sur-Yvette, France}

\author{Jean-Daniel Bancal}
 \affiliation{Institute of Theoretical Physics, University of Paris-Saclay, CEA, CNRS, Gif-sur-Yvette, France}

\author{Andreas Wallraff}
 \affiliation{Department of Physics, ETH Zurich, 8093 Zurich, Switzerland}
 \affiliation{Quantum Center, ETH Zurich, 8093 Zurich, Switzerland}

\date{\today}

\begin{abstract}
Self-testing protocols enable the certification of quantum systems in a device-independent manner, i.e.~without knowledge of the inner workings of the quantum devices under test. Here, we demonstrate this high standard for characterization routines with superconducting circuits, a prime platform for building large-scale quantum computing systems. We first develop the missing theory allowing for the self-testing of Pauli measurements. We then self-test Bell pair generation \textit{and} measurements at the same time, performing a \textit{complete} self-test in a system composed of two entangled superconducting circuits operated at a separation of 30 meters. In an experiment based on 17 million trials, we measure an average CHSH (Clauser-Horne-Shimony-Holt) S-value of 2.236. Without relying on additional assumptions on the experimental setup, we certify an average Bell state fidelity of at least $58.9\%$ and an average measurement fidelity of at least $89.5\%$ in a device-independent manner, both with 99$\%$ confidence. This enables applications in the field of distributed quantum computing and communication with superconducting circuits, such as delegated quantum computing.
\end{abstract}

\maketitle

Many everyday applications rely on testing routines to verify the proper functioning of the system. This is particularly true for the development of information processing systems \cite{Agrawal1994}. With the maturing of quantum information processing units based on superconducting circuits \cite{Krinner2022,Acharya2023,Arute2019,Kim2023b}, neutral atoms \cite{Bluvstein2023} or trapped ions \cite{Moses2023}, verifying the correct operation of such systems has become crucial as well.

Preferably, the verification process is based on a minimal set of assumptions about the inner workings of the device under test. The growing complexity of quantum systems however imposes substantial challenges in verifying the correct working of the implemented operations. Most conventional verification schemes have the shortcoming that they certify specific aspects of a quantum system while making detailed assumptions about the correct working of other critical parts of the setup. This is not ideal from a security point of view. 
Self-testing allows one to certify quantum systems in a \emph{device-independent} way, i.e.~free of any assumption about the inner workings of both the tested device and any additional quantum devices used during the test~\cite{Mayers1998,Mayers2004,Acin2007b}, see Refs.~\cite{Scarani2012, Supic2020} for a review. Beyond the verification of elementary quantum operations such as quantum state preparation~\cite{Coladangelo2017}, measurements and gates~\cite{Wagner20,Magniez2006,Sekatski2018}, self-testing also plays a key role in certifying quantum protocols such as delegated quantum computing~\cite{Fitzsimons2017,Gheorghiu2018}.

Self-testing relies on quantum non-locality \cite{Brunner2014}, asserted through the violation of a Bell inequality~\cite{Bell1964, Bell2004} between at least two quantum information processing nodes, see Fig.~\ref{fig:selftest_scheme}a. As for a Bell test, the protocol repeats $n$ identical trials, where in each trial a verifier feeds the untrusted device under test with input bits $x$ and $y$ of their choice, and records the corresponding measurement outcomes $a$ and $b$. The core idea of self-testing is that observing a large Bell inequality violation implies that the system at hand operates in a manner close to an ideal quantum model. As such, self-testing relies on correlations between quantum systems which cannot be explained by any classical model. Self-testing thereby provides a fundamentally different validation from conventional methods, not possible with classical means only.

Since classical models with post-processing can mimic quantum correlations~\cite{Larsson2014}, the detection loophole must be closed in order to run a self-testing protocol, so that the fair-sampling assumption does not have to be employed. Similarly, the nodes should not be able to exchange classical information during the test, so that classical models leveraging communication cannot constitute an explanation for the observed Bell inequality violation~\cite{Toner2003}. This is also known as closing the locality loophole. When these loopholes are closed, it becomes possible to infer the quality of the quantum state shared between the remote qubits in a black-box manner. This means that the certification can be done solely from the correlations observed in each trial between the measurement settings chosen by the user and the corresponding measurement outcomes~\cite{Supic2020}, and therefore without assuming any specific description of the device.

\begin{figure*}
    \includegraphics[width=\linewidth]{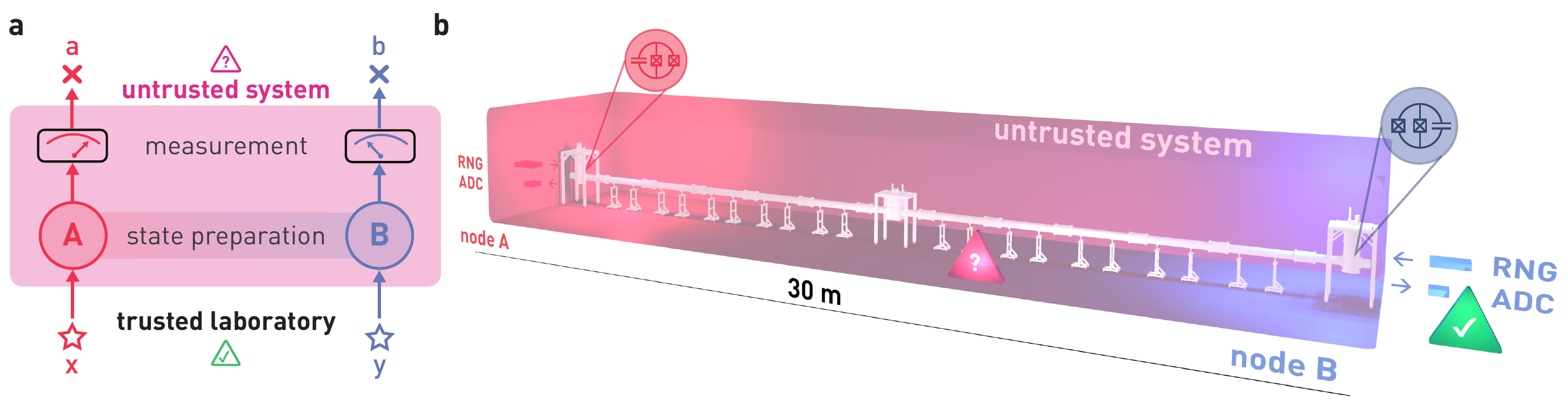} 
    \caption{\textbf{Experimental setup.} a, Schematic representation of a self-testing scheme with space-like separation. In each trial, A and B choose the measurement bases x and y using trusted random number generators, and the two parties record the measurement outcomes a and b. These start and stop events at A and B are space-like separated to prohibit classical communication among the nodes in each trial, closing the locality loophole. The untrusted setup (magenta) consists of all devices relevant for generating a Bell state and for performing a Bell test. This particularly includes the quantum channel between the nodes and all components relevant for the measurement process up to the moment $t_\times$ when the measurement result is said to be fully determined (see Appendix~\ref{app:timing}). b, Schematic illustration of the experimental setup. At the two nodes A and B each, a superconducting circuit qubit is operated in a dilution refrigerator. The two nodes are remotely connected through a 30-meter-long microwave waveguide to generate entanglement between the two qubits, see main text. Each node operates a trusted random number generator (RNG) and measurement signal detection device (analog-to-digital converter, ADC). The untrusted setup, represented in a magenta box, is marked with a magenta triangle, the trusted part with a green triangle.}
    \label{fig:selftest_scheme}
\end{figure*}

Most noise-tolerant self-testing protocols known today impose more stringent requirements on the experimental setup than is required to perform a Bell test. Specifically, the most noise-robust analysis known to date \cite{Kaniewski2016} requires a CHSH \cite{CHSH1969} S-value larger than $S^\star=(16+14\sqrt{2})/17\simeq 2.106$, and a large number of Bell test trials $n\gtrsim10^7$ to avoid the conclusion to be affected by finite size effects of the statistical sample (Appendix~\ref{app:statistics}). Moreover, closing the detection loophole and ensuring no communication between the nodes typically further increases the complexity of the experimental setup and constrains the quality of the relevant metrics, e.g. through limitations on the duration of the measurement process.

Starting in 2018, the first proof-of-principle self-tests were demonstrated, all relying on the fair sampling assumption. Specifically, the self-testing relation between Bell inequality violation and state fidelity was first used in an experiment with bipartite and tripartite states of optical photons~\cite{Zhang2018p}. Multidimensional entangled photonic states were then investigated with a table-top~\cite{Zhang2019r} and an integrated optics system~\cite{Wang2018o}. Table-top optical setups were also used to study partially entangled states of photons and to compare the self-testing fidelity bounds to device-dependent ones~\cite{Gomez2019,Goh2019}. Furthermore, sharing two copies of a quantum state in a bipartite and tripartite optical system was studied~\cite{Agresti2021}, and self-testing for photonic graph states was explored~\cite{Wu2021a,Xu2022f}.

The detection loophole was first closed in a self-testing experiment using two ions in a single trap~\cite{Tan2017}. Four-qubit GHZ states on a single superconducting circuit device later also closed the detection loophole~\cite{Wu2022a}. In all these cases, ensuring the absence of communication between the nodes relied on physical isolation. The latter work \cite{Wu2022a} also included a photonic experiment where fair sampling is assumed but the absence of communication is supported by space-like separation. The only self-testing analysis so far including both space-like separation of the nodes and closing the detection loophole was obtained in the aposteriori analysis presented in Ref.~\cite{Bancal2021} of the loophole-free Bell test in Ref.~\cite{Rosenfeld2017} using neutral atoms. However, because of the low repetition rates of that experimental setup (about 30~mHz), this analysis required to use data captured over the course of several months. This is inconvenient for practical purposes where self-testing is considered as a subroutine of a quantum computing algorithm with a total run time on the order of minutes to hours.

Furthermore, only the self-testing of quantum states has been demonstrated experimentally so far. However, the certification of \textit{both} the states and the measurements~\cite{Guhne2023} through self-testing is crucial to certify the blindness and verifiability of delegated quantum computing with entangled provers~\cite{Gheorghiu2015, Hajudesk2015}.
Reasons for not trusting the measurement apparatus include potential malfunctions, miscalibrations, biased or tampered results, or the device classically emulating a quantum measurement. 

In this article, we present a \emph{complete self-test} of a system of two superconducting circuits remotely entangled through a 30-meter-long microwave quantum link. We refer to the procedure as \textit{complete} since it self-tests at the same time a Bell state between the two distant circuits, and a pair of Pauli measurements within each device.
Our demonstration closes the detection loophole and ensures the absence of communication between the nodes through space-like separation. Enabled by a unique combination of high Bell inequality violations and high data acquisition rates (Appendix~\ref{app:literature}), our experiment achieves the targeted goal in a total time of only about half an hour. This makes it useful for practical implementations.
Our experiments demonstrate the utility of this verification scheme in local area networks of remote superconducting quantum processors, with prospects in secure communication and distributed computing. 

\section*{Experimental System and Protocol}

Our two-node system (Fig.~\ref{fig:selftest_scheme}b) consists of two transmon-style superconducting qubits operated in their individual dilution refrigerators at temperatures of about 15~mK, as described in earlier reports \cite{Magnard2020, Storz2023}. We control and read out the state of the qubits using microwave pulses on the nanosecond time scale. The two qubits are connected through a 30-meter-long cooled microwave quantum channel, a rectangular waveguide, used for the distribution of entanglement \cite{Kurpiers2017, Kurpiers2018}. The space-like separation of the two qubits at a linear distance of about 30~meters provides a window of about 100~ns for each trial of the experiment, from the generation of the input bits $x$ and $y$ to obtaining the measurement outcomes $a$ and $b$, during which the exchange of information between the two nodes is precluded by the laws of special relativity (Fig.~\ref{fig:protocol}). This ensures that no communication occurs between the nodes during the experimental trial, closing the locality loophole. 
\begin{figure}
  \begin{center} \includegraphics[width=1\columnwidth]{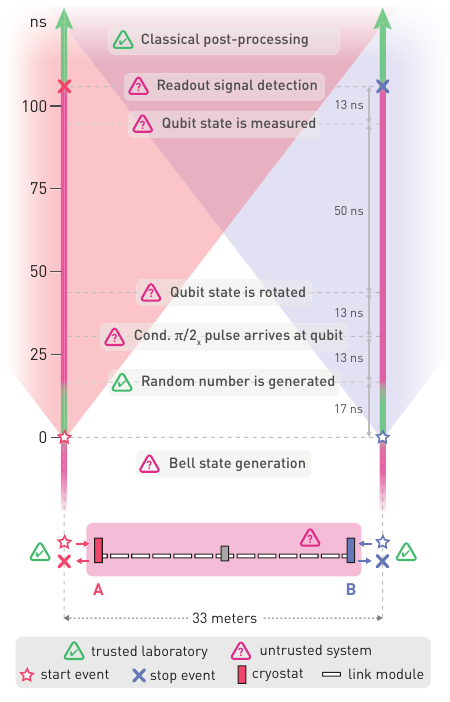}
  \caption{\textbf{Self-testing protocol.} Space-time configuration of a single trial of the experiment. The colors and triangles indicate whether the step of the trial, or part of the setup, belongs to the trusted (green) or untrusted (magenta) part. The individual steps take place simultaneously at both nodes. Stars and crosses mark the start and stop events of each trial, and the shaded red and blue regions indicate the forward light cones of the corresponding start events.}
  \label{fig:protocol}
  \end{center}
\end{figure}

In each trial of the experiment (Fig.~\ref{fig:protocol}), we first establish entanglement between the two remote qubits by exchanging a single microwave photon through the quantum channel \cite{Kurpiers2018}, yielding a (readout-error-corrected) Bell state fidelity of 85.9$\%$, as characterized by quantum state tomography. Once the two qubits are entangled, we randomly choose the measurement bases of both qubits (Fig.~\ref{fig:protocol}). Those are determined by bits provided by trusted random number generators \cite{Abellan2015} located close to each node of the untrusted setup, to support the assumption of measurement independence \cite{Hall2010}. At each site, the local random bit controls the untrusted rotation of the qubit state through the conditional application of a microwave pulse, which implements the basis choice selection. Finally, we perform an untrusted single-shot measurement of the state of the two individual qubits \cite{Walter2017} (Fig.~\ref{fig:protocol}) and record the result in a trusted memory. We record and analyze the result of all trials, which closes the detection loophole. 

With this setup, we reach CHSH S-values exceeding $S_{\mathrm{avg}}=2.2$, well above the aforementioned critical threshold $S_{\mathrm{avg}}>S^*$, and therefore sufficiently high for enabling a self-testing analysis. This marks a considerable improvement over previous work on the same setup (Ref.~\cite{Storz2023}, $S_{\mathrm{avg}}=2.074$), which we attribute to a reduction of the loss in the quantum channel, see Appendix~\ref{app:updates}.

In Bell-like experiments, the finite size of the statistical dataset typically leads to a substantial reduction of the S-value the experiment guarantees (Appendix~\ref{app:statistics}). In order not to be significantly affected by such finite-size effects of the statistical data set, we record more than 16 million trials. This constitutes more than a 10-fold increase over our previous work \cite{Storz2023}, enabled by an improved temporal stability of the calibrations of the experimental setup (Appendix~\ref{app:updates}).

\section*{Results}
For the data presented here, we ran $n$~=~$2^{24}$~=~16'777'216 trials of the experiment at a repetition rate of 50~kHz. After each $2^{20}$~=~1'048'576 trials, we recalibrated the single-shot assignment thresholds. In total, including software overhead, the experiment took about 40 minutes. The calibrations of the microwave pulses, timings and readout parameters performed before the experiment are summarized in Appendix~\ref{app:calibrations}. 

We find an average CHSH S-value of $S_{\mathrm{avg}}=2.236$ for the data taken in all $n$ trials (Figure~\ref{fig:S_vs_time}). 
In a fully device-independent context, the fidelity between the physical state used in a Bell-like experiment and the singlet state $\ket{\phi^+}=(\ket{00}+\ket{11})/\sqrt{2}$ can be lower bounded from the CHSH S-value using~\cite{Kaniewski2016}
\begin{equation}
    \mathcal{F}_{\text{s}}(S) \geq \frac{1}{2}+\frac{1}{2} \cdot \frac{S-S^*}{2 \sqrt{2}-S^*}.
    \label{eq:selftesting}
\end{equation}
This bound ensures the existence of two-dimensional Hilbert subspaces in A and B in which the state has a fidelity of at least $\mathcal{F}_\text{s}(S)$ with the singlet, independently of the dimension of the full physical state and of the measurement description (see Appendix~\ref{app:theory_main}). Taking into account finite size statistics on the S-value (see Appendix~\ref{app:statistics}), we can guarantee a minimal state fidelity of $\mathcal{F}_\text{s} = 58.9\%$ at a $99\%$ confidence level, in a device-independent manner.

In order to bound the measurement fidelity in the presence of noise or experimental imperfections, we model the measurement apparatus of A(B) as a device taking a setting choice $x (y)$ and a quantum state $\rho$ as input, and producing an output $a (b)$. We then compare it to the expected pair of Pauli measurements by choosing a suitable injection map~\cite{Sekatski2018}. In Appendix~\ref{app:theory_main} we show that this allows us to guarantee that the tested apparatus performs as the expected one up to a fidelity that is lower bounded by the S-value as
\begin{equation}
    \mathcal{F}_\mathrm{m} (S) \geq \frac{ \sqrt{2} \, S +4}{8}.
    \label{eq:fid_bound}
\end{equation}
In other words, the measurement fidelity must be at least $\mathcal{F}_\mathrm{m}(S)$ in order to observe a CHSH value $S$, even at unit entangled state fidelity, and even if the measurements act on different Hilbert spaces than expected. Taking finite size statistics into account, this bound allows us to device-independently verify a minimal measurement fidelity of $\mathcal{F}_{\mathrm{m}}$~=~89.5$\%$, also at a 99$\%$ confidence level.

For comparison, we note that doing quantum tomography of the state without readout error correction leads to a Bell state fidelity of $\mathcal{F}_{\text{s}}=83.9\%$. This procedure is not device-independent as it assumes a two-qubit model and trusted measurements. In Appendix~\ref{app:measurement}, we also estimate the fidelity of the measurement apparatus in a way that is not device-independent by combining the process fidelities of the one-qubit rotations and of the readout, leading to $\mathcal{F}_{\text{m}}\geq97.2\%$. The difference between these fidelities and the ones obtained in a device-independent fashion through self-testing can be explained by the distance of the observed CHSH value from the ideal CHSH value~\cite{Goh2019}.

\begin{figure}
  \begin{center} \includegraphics[width=1\columnwidth]{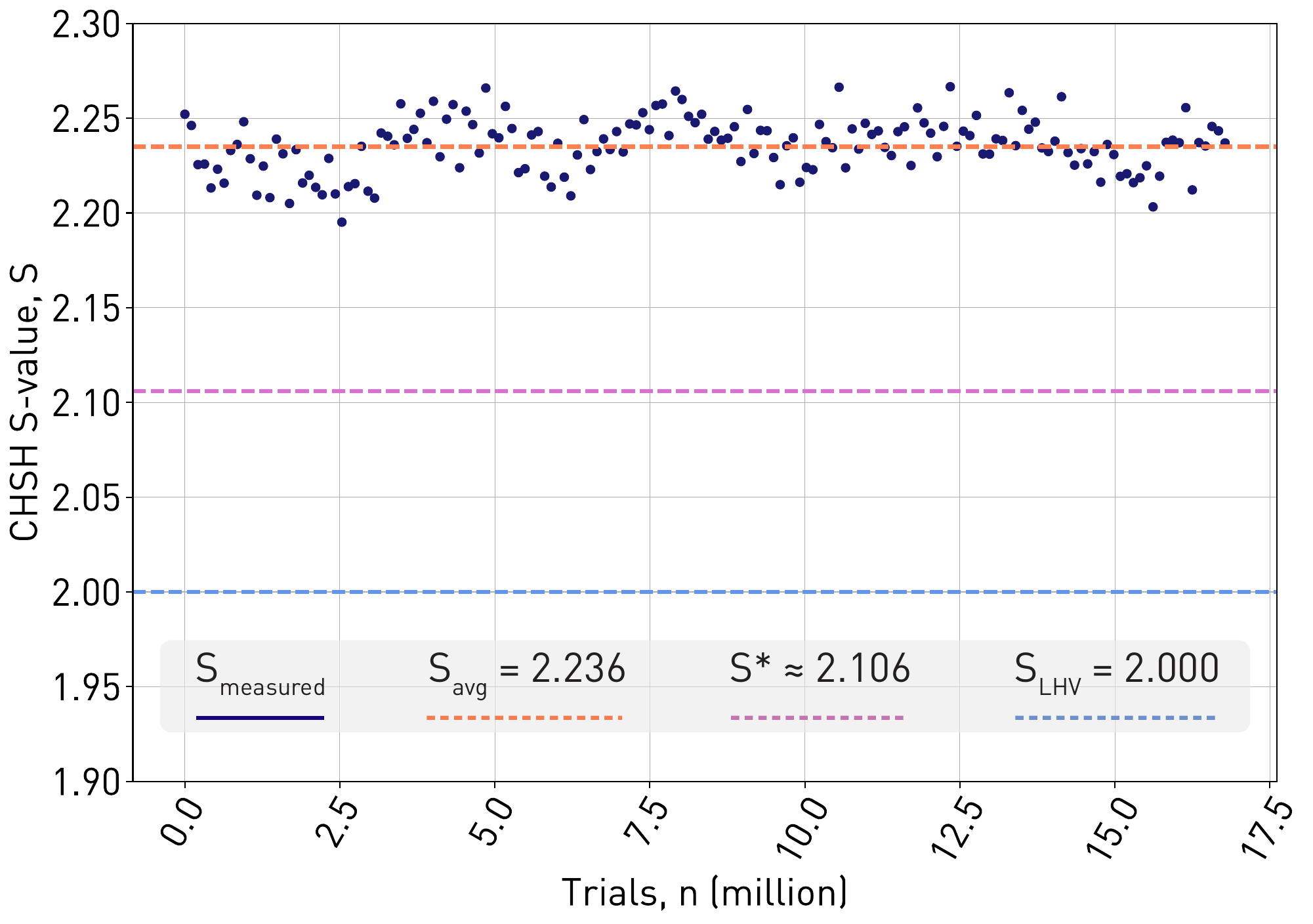}
  \caption{\textbf{Self-testing; evolution of the CHSH S-value during the experiment.} Recorded CHSH S-value during the full experiment, i.e.~all $n=2^{24}\approx16.7$~million trials. The blue dots $S_\mathrm{measured}$ show the experimental, average S-values of $n/128=2^{17}=131'072$ trials each. The dashed lines represent the average measured S-value $S_\mathrm{avg}$ of the full dataset, and the critical thresholds for self-testing the state ($S^*$) and the measurement ($S_\mathrm{LHV}$; bound imposed by local hidden-variable models).}
  \label{fig:S_vs_time}
  \end{center}
\end{figure} 

The device-independent nature of the experiment is maintained by evaluating each and every trial, which closes the detection loophole, and by not assuming the trials to be independent and identically distributed, addressing the memory loophole \cite{Larsson2014}. Furthermore, we close the locality loophole by asserting space-like separation of the start and stop events of each trial. To verify the latter, we performed precise measurements of the space-time configuration of our experimental setup (Appendix~\ref{app:timing}).

\section*{Discussion of Performance}

The presented experiment constitutes the first implementation of a complete self-testing that bounds the quality of a distributed entangled state and of the measurement apparatus at the same time. It also is the first self-testing implementation with superconducting circuits including space-like separation between the measurements.
The experiment further reaches previously unmatched fidelities certified in a device-independent manner (Appendix~\ref{app:literature}), i.e. without extra assumptions on the experimental setup, as in Ref.~\cite{Bancal2021}. Finally, it achieves this goal in a total run time of only tens of minutes.

The certifiable fidelities on this setup could be increased further by enhancing the observable S-value. This could be achieved by extending the lifetime and coherence time of the qubits, and by improving the fidelity of the distributed entangled state by lowering the photon loss in the quantum channel. One strategy in this regard is to only use superconducting materials for the interconnects between the quantum devices and the waveguide, as in \cite{Zhong2019}, or to use a heralded entanglement generation scheme \cite{Kurpiers2019} at the cost of a lower effective repetition rate of the experiment. Another approach is to increase the readout fidelity further, e.g.~through the addition of a more optimized device design \cite{Swiadek2023}, an on-chip tunable Purcell filter \cite{Sunada2024}, or a pre-arming of the readout resonator with photons \cite{Lledo2023}. Unit state and measurement fidelities can only be certified with a perfect S-value (Appendix~\ref{app:theory_main}).

We envision that, once the quality of interlinked networks between remote superconducting circuit qubits increases further, self-testing may serve as a standard certification routine to test the proper functioning of a system in a conceptually simple yet effective manner, either as an addition or an alternative to established certification routines such as randomized benchmarking \cite{Knill2008} or tomography methods \cite{Chuang1997}. The scheme can be particularly relevant in the context of blind quantum computing, where the experimental setup is not, or not fully trusted. Furthermore, self-testing is at the heart of other protocols relevant for certifying the security of quantum networks. Potential applications include device-independent quantum key distribution \cite{Nadlinger2022,Zhang22,Liu22}, conference key agreement \cite{Murta2020}, or delegated quantum computing \cite{Gheorghiu2019}.
\\

We thank Reto Schlatter for his assistance with the operation of the cryogenic setup, Raul Conchello Vendrell and Christoph Hellings for valuable contributions to the software framework, and Jean-Claude Besse and Kevin Reuer for helpful discussions. The work at ETH Zurich was funded by the European Union's Horizon 2020 FET-Open project SuperQuLAN (899354) and by ETH Zurich. The team in Paris-Saclay acknowledges funding by the European High-Performance Computing Joint Undertaking (JU) under grant agreement No 101018180 and project name HPCQS, the European Union’s Horizon Europe research and innovation program under the project “Quantum Security Networks Partnership” (QSNP, Grant Agreement No.~101114043) and by a French national quantum initiative managed by Agence Nationale de la Recherche in the framework of France 2030 with the reference ANR-22-PETQ-0009.

\appendix

\section{Effect of Finite Sample Size}
\label{app:statistics}
In this section, we discuss the influence of the size of the statistical dataset, i.e.~the number of trials $n$, on the certifiable fidelities in the experiment.

The certification schemes for the shared Bell state, detailed in Ref.~\cite{Kaniewski2016}, and for the measurement devices, discussed in Appendix~\ref{app:theory_main}, provide a mapping between the CHSH S-value and the certified fidelities, see Eqs.~\eqref{eq:selftesting} and \eqref{eq:fid_bound} in the main text. While these formulas are valid for all values of $S$, in practice the S-value that characterizes a setup is never known exactly, and it may even evolve during the course of an experiment, for example due to drifts in the calibration parameters. Therefore, a statistical analysis is needed to assess the S-value that can be supported by a finite number of experimental samples.

In order to account for potential parameter fluctuations along the course of the experiment, we assign a distinct S-value $S_i$ to each trial $i$. The average S-value over the course of the experiment is then given by
\begin{equation}
S_\mathrm{avg}=\frac{1}{n}\sum_i S_i.
\end{equation}
To take into account finite size effects, we are interested in computing a lower bound $S_{\mathrm{avg},\alpha,n}\leq S_\mathrm{avg}$ on the average S-value $S_\mathrm{avg}$ for all possible realizations $\{S_i\}_i$ of $n$ trials with probability at least $1-\alpha$, i.e.~with the guarantee that
\begin{equation}
P(S_{\mathrm{avg},\alpha,n} \leq S_\mathrm{avg}) \geq 1-\alpha.
\end{equation}
This defines a one-sided confidence interval on $S_\mathrm{avg}$. By the convexity of Eqs.~\eqref{eq:selftesting} and \eqref{eq:fid_bound} in the main text, this then implies one-sided confidence invervals with lower bounds $\mathcal{F}_\mathrm{s}(S_{\mathrm{avg},\alpha,n})$ and $\mathcal{F}_\mathrm{m}(S_{\mathrm{avg},\alpha,n})$ on the fidelity of the state and of the measurements, each with confidence $1-\alpha$.

We calculate the finite-size-corrected S-value (for $n$ trials) by employing the method of Ref.~\cite{Bancal2022b}, which provides a tight lower bound on the average success probability of $n$ Bernoulli trials. Importantly, this method does not assume that the $n$ trials are independent and identically distributed (IID), thereby closing the memory loophole \cite{Larsson2014}. In this context, we introduce the Bell test in the framework of a game. As described further in Ref.~\cite{Storz2023}, in such a setting the two parties A and B aim to win as a team. The game consists of $n$ rounds. In each trial $i$, A and B individually receive a question, corresponding to a random and uniform input bit $x$ and $y$, and they are asked to reply with a response, corresponding to the measurement outcomes $a$ and $b$. A and B are said to win a certain round of the game if $ \textit{x} \wedge \textit{y} = \textit{a} \bigoplus \textit{b} $, where $\wedge$ denotes the logical AND function and $\bigoplus$ the XOR operation. The winning probability for round $i$ is given by $p_i=(4+S_i)/8$. The best possible classical strategy allows A and B to achieve a success probability of $p_{\mathrm{win}}^{\mathrm{LHV}} = 3/4 = 0.75$, but a strategy involving shared entanglement achieves higher values, up to $p_{\mathrm{win}}^{\mathrm{QM}} = \cos^2{(\pi/8)}\approx0.854$.

\begin{figure}
  \begin{center} \includegraphics[width=0.9\columnwidth]{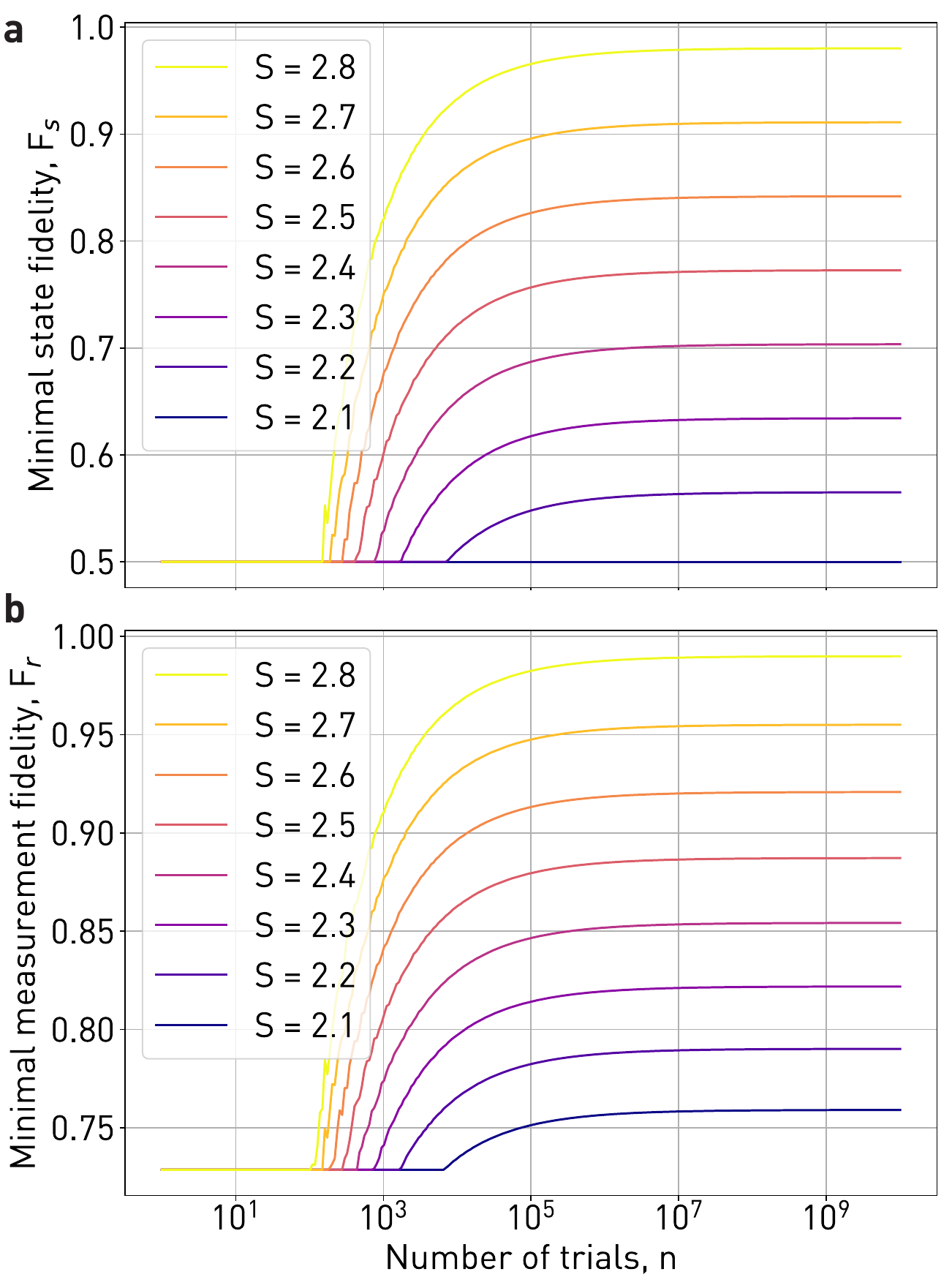}
  \caption{\textbf{Effect of finite statistics on the self-testing results.} Dependence of the minimal certifiable state (\textbf{a}) and measurement (\textbf{b}) fidelities given an observed S-value $S$ and the number of experimental trials $n$ in a self-testing experiment.}
  \label{fig:n_behavior}
  \end{center}
\end{figure}

After all $n$ trials, we calculate the number of times $c$ where A and B won the game. The average winning probability given $n$ trials and confidence $1-\alpha$ is lower bounded by \cite{Bancal2022b}
\begin{equation}
p_{\mathrm{avg},\alpha,n}=
\begin{cases}
0 & c=0\\
I^{-1}_\alpha(c,n-c+1) & c>1,\ \alpha\leq\alpha^* \\
\frac{1}{n}\left(c - \frac{1-\alpha}{1-\alpha^*}\right)& c>1,\ \alpha > \alpha^*
\end{cases},
\end{equation}
where $\alpha^*=I_{(c-1)/n}(c,n-c+1)$. Here, $I$ is the regularized incomplete Beta function and $I^{-1}$ its inverse. When $\alpha\leq1/4$, this bound is tight and the corresponding CHSH S-value is given by
\begin{equation}
S_\mathrm{avg,\alpha,n} = 8p_{\mathrm{avg},\alpha,n}-4.
\end{equation}
As a result, small sample sizes yield a substantial reduction of the finite-size-corrected S-value compared to the measured value, and thereby lead to the certified fidelities being lower. Figure~\ref{fig:n_behavior}a displays this behavior for a set of measured S-values for the certification of a minimal Bell state fidelity and Fig.~\ref{fig:n_behavior}b for the minimal average measurement fidelity. We observe that finite-size effects substantially matter for small sample sizes $n\ll10^7$, whereas for $n\geq10^7$ they become negligible. For this reason, we choose to measure $n=2^{24}\approx1.67*10^7$ trials in the presented self-testing experiment.

\begin{figure*}
  \begin{center} \includegraphics[width=1.7\columnwidth]{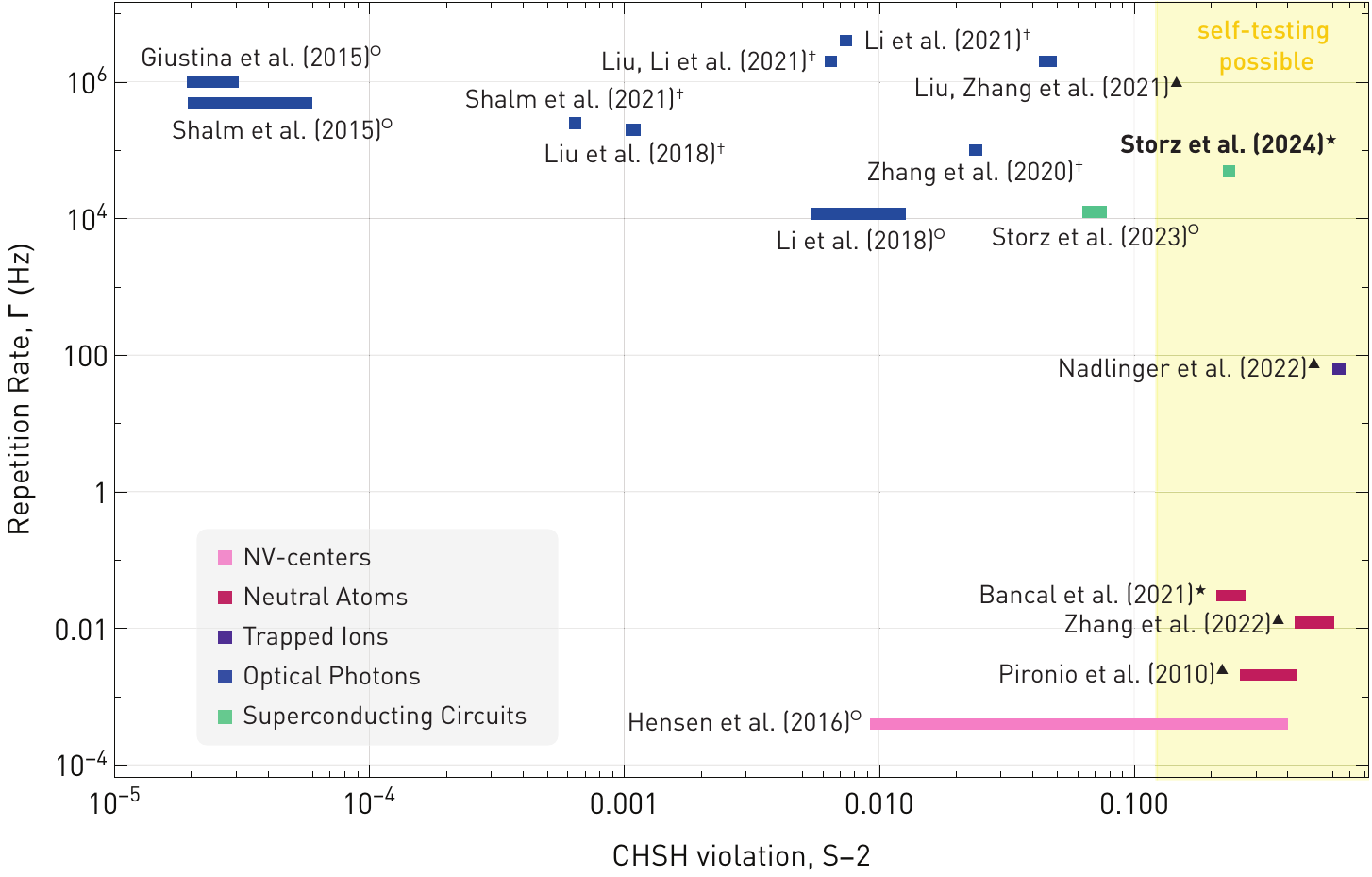}
  \caption{\textbf{Comparison of related experiments.} This figure compares the CHSH inequality violation and the repetition rate of published device-independent quantum information processing experiments and loophole-free Bell tests. For each experiment, we display a horizontal bar where the experimentally measured S-value marks the high end (i.e. the right side), and the lower (left) end represents the achieved S-value, with 99$\%$ confidence, when including finite size effects (see Appendix~\ref{app:statistics}). Typically, experiments involving more statistics are represented by a shorter bar. Only if the results of an experiment fully fall into the yellow shaded region, a self-testing analysis is possible even when including finite-size effects. We note that three of the five experiments which fulfill that criterion, i.e.~the quantum key distribution implementations in Refs.~\cite{Nadlinger2022,Zhang22} and the pioneering randomness generation experiment in Ref.~\cite{Pironio2010}, do not exclude classical communication between the nodes by space-like separation, and instead rely on additional assumptions on shielding properties of their devices. The colors of the markers represent the experimental platform, see legend, and the superscript symbols describe the type of experiment: loophole-free Bell tests (circles), device-independent randomness generation and expansion (daggers), device-independent quantum key distribution (triangles), and device-independent self-testing (stars). 
  }
  \label{fig:comparison}
  \end{center}
\end{figure*}

\section{Comparison to relevant literature}
\label{app:literature}
In this section, we compare our experiment to other implementations of device-independent quantum information processing algorithms, and to related experimental loophole-free Bell tests.

For most of these protocols, it is critical to achieve high Bell inequality violations (i.e. high CHSH S-values), while being able to capture a large amount of statistics (i.e.~trials $n$) over the course of a short time, as discussed in the main text and in Appendix~\ref{app:statistics}. For self-testing specifically, the corresponding analysis is only possible with a minimal S-value of $S>S^*=2.106$ \cite{Kaniewski2016}, and with preferably at least $n=10^7$ trials, thereby requiring a challenging combination of both of these metrics. 

Figure~\ref{fig:comparison} compares earlier relevant experiments in terms of the achieved CHSH S-value, and the number of trials run. Optical setups are typically able to collect a large number of trials $n$ in a short time, but they fall short in achieving high S-values. On the other hand, solid-state based implementations typically achieve high Bell inequality violations but struggle with the acquisition of a large number of trials. For these reasons, it is hard to reach the targeted regime (yellow shaded area) that enables device-independent self-testing. The experimental setup presented in this publication reaches a useful regime with both high S-values and large datasets, enabling self-testing with high repetition rates. 
As discussed in Ref.~\cite{Bancal2021}, the neutral atom setup presented in \cite{Rosenfeld2017} only manages to collect enough statistics when including data acquired over the course of several months.

We further note that the listed experiments performing device-independent quantum key distribution (triangle superscripts in Figure~\ref{fig:comparison}) did not close the locality loophole by space-like separation of the nodes, and therefore do not technically fulfill the corresponding requirement for self-testing. 
In our system, actively leaving the locality loophole open allows for longer readout integration times, thereby increasing the readout fidelity and so the achieved S-value (see Refs.~\cite{Storz2023, Magnard2021}). In such a configuration, we expect to reach S-values above $S=2.3$.

\section{Enabling Updates to the Experimental Setup}
\label{app:updates}
In the following, we briefly discuss the key changes to the experimental setup originally presented in Ref.~\cite{Storz2023}, allowing us to increase the S-value from $S=2.074$ to $S=2.235$, a critical step which enabled the self-testing experiment presented here.

The S-value is predominantly limited by the Bell state fidelity of the two qubits, and that metric in turn by the photon loss in the quantum channel. We therefore focused on reducing the photon loss in the quantum channel between the two nodes by replacing some of the components forming the interconnects between the quantum devices, printed circuit boards and the waveguide, with less lossy ones. Specifically, we have replaced copper-based coaxial microwave cables, which connect the quantum devices to the waveguide, with superconducting coaxial cables out of Niobium-Titanium. In addition, we have removed a cryogenic circulator in the channel, which was previously used to characterize the mode function of the transmitted photons, and to aid with the thermalisation of the stray radiation in the waveguide. Moreover, we have reduced the number of adapters in the quantum channel. These changes lowered the photon loss in the channel in total by about 6$\%$. 

Furthermore, in the presented self-testing experiment, we aimed to take an order of magnitude more data compared to the loophole-free Bell test published earlier \cite{Storz2023}. This required substantial improvements in the electronic synchronisation of the different devices of the setup, so that the calibration of pulses, readout weights and other aspects remained stable over extended periods of time. For further details on this aspect, we refer to Ref.~\cite{Storz2023a}.

\section{Pulse Calibrations}
\label{app:calibrations}
In this section, we discuss calibration measurements performed before the self-testing experiment, with the purpose of characterizing the performance of the individual components of the system.

\paragraph{Single-shot readout}
First, we calibrate the integration weights of the FPGA for the employed rapid, single-shot readout scheme \cite{Walter2017}. For this purpose, we prepare each qubit 8000 times in the ground and excited state, respectively, and record both quadratures of the measurement signal. Using this data, we determine the optimal state assignment threshold, and the corresponding integration weights for a total integration time of 50~ns. During the main experiment, we repeat this calibration process after each set of $2^{20}$ trials of the self-testing experiment, in order to capture and then mitigate phase drifts of the devices involved in the readout scheme, which cause drifts of the state assignment threshold value and therefore a reduced readout fidelity. The readout fidelities extracted from these 16 calibration iterations remain mostly stable over time, with an average readout fidelity of $F_{r,A}=1-p(e|g)-p(g|e)=98.9\%$ for qubit A and $F_{r,B}=97.2\%$ for B, as shown in Figure~\ref{fig:readout}.

\paragraph{Entanglement generation}
In the following, we discuss the calibration of the pulses involved in the entanglement generation scheme. For distributing entanglement between the superconducting qubits at the two nodes, we employ the scheme of Ref.~\cite{Kurpiers2018}. Specifically, we prepare qubit A in an equal superposition state between the first ($\ket{e}$) and the second ($\ket{f}$) excited state, and then drive the $\ket{f0}\leftrightarrow\ket{g1}$ sideband transition of a coupled qubit (first state index) and on-chip resonator (second state index) system using a microwave pulse. This pulse transfers the excitation of the qubit into a microwave resonator in the form of a photon, which in turn decays into the waveguide. We calibrate the time-dependent amplitude and phase of this emission pulse to have qubit A emit a photon at a rate $\Gamma\approx20$~MHz \cite{Magnard2018}. Upon arrival at node B, we reabsorb the microwave photon in qubit B using the time-reversed process, at the same rate. Afterwards, we map the $gf$-manifold of qubit B to the $ge$-subspace. Finally, we perform a local unitary $\pi/2$ rotation around the x-axis of qubit A and around the (x+$\Theta$) axis on qubit B, where the angle $\Theta$ is experimentally calibrated to compensate for a synchronisation offset between the electronics of the two setups. Ideally, this scheme yields a Bell state $\ket{\phi}^+=(\ket{gg}+\ket{ee})/\sqrt{2}$.

It is further important to calibrate the time at which node B starts reabsorbing the photon sent through the quantum channel. It takes about 155~ns for the photon to travel from node A to B, so calibrating the start of the absorption pulse at B accordingly ensures that the majority of the energy of the photon is transferred to the qubit. This optimal time is determined by emitting a photon at node A, and reabsorbing it at node B while changing the start time of the absorption pulse. After a photon transfer from A to B, we measure the state of the qubit using averaged, 3-level readout. Ideally, qubit B ends up in the f-level with population $p_f=1$. However, caused by photon loss in the quantum channel, qubit decay, and imperfections in the emission and absorption process, this value is lowered to about $p_f=83\%$. This characterizes the transfer fidelity. We find the optimal delay between the emission and the absorption pulse to be $\tau=155.4$~ns, which aligns with the expected photon traveling time through the waveguide and the synchronisation offset between the arbitrary waveform generators responsible for generating the pulses at the two nodes.

We then perform quantum state tomography of the generated Bell state, yielding a fidelity of 85.9~$\%$ on average when correcting for readout errors, and 83.9~$\%$ otherwise.

\begin{figure}
  \begin{center} \includegraphics[width=\columnwidth]{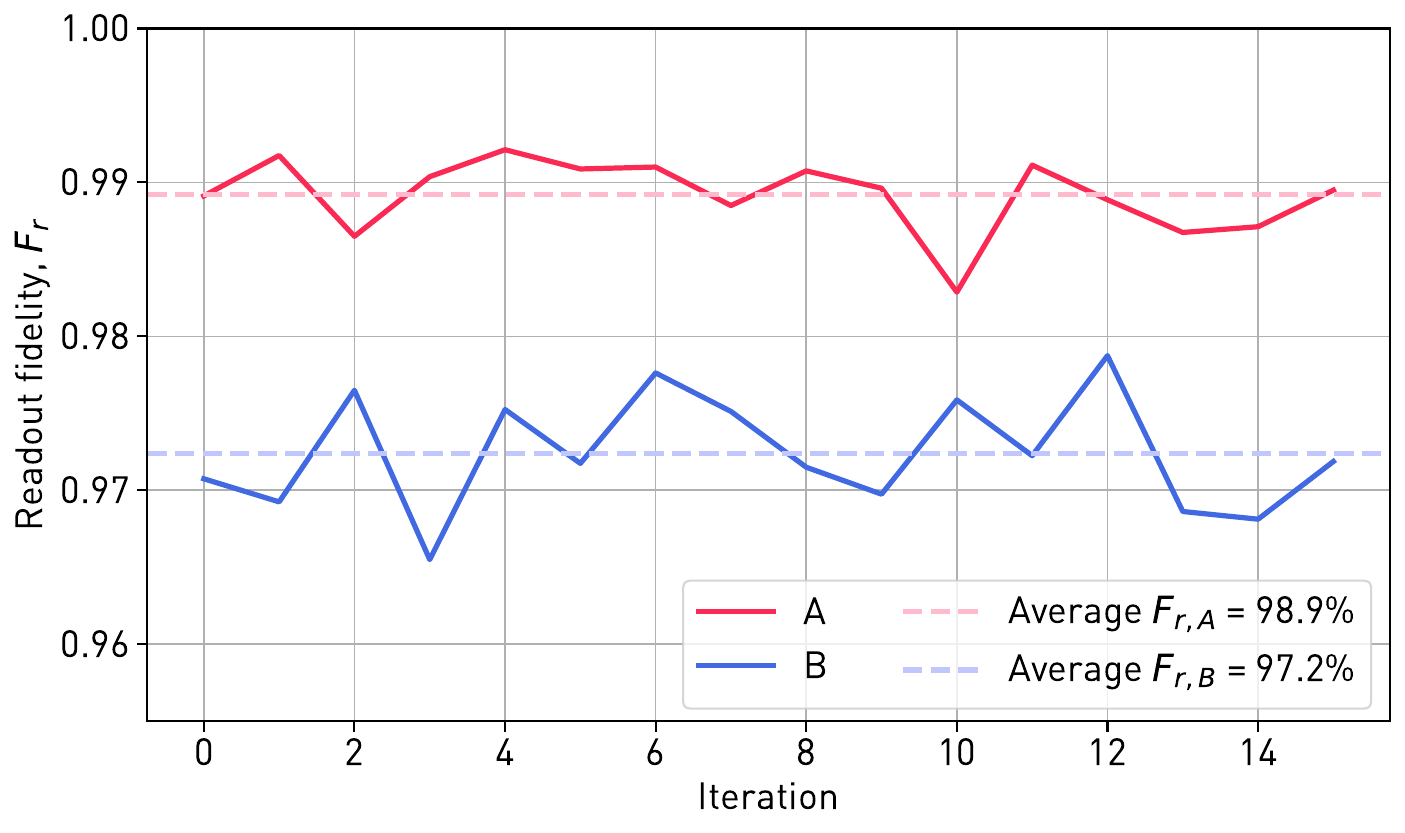}
  \caption{\textbf{Evolution of the measurement fidelity over time.} Measured single-shot readout fidelity on nodes A and B during each of the 16 consecutive acquisition blocks of the self-testing experiment, each comprising $n=2^{20}$ trials and taking approximately 150 seconds (including software overhead). The readout integration duration is 50~ns on both nodes.}
  \label{fig:readout}
  \end{center}
\end{figure}

\begin{figure*}
  \begin{center} \includegraphics[width=1.9 \columnwidth]{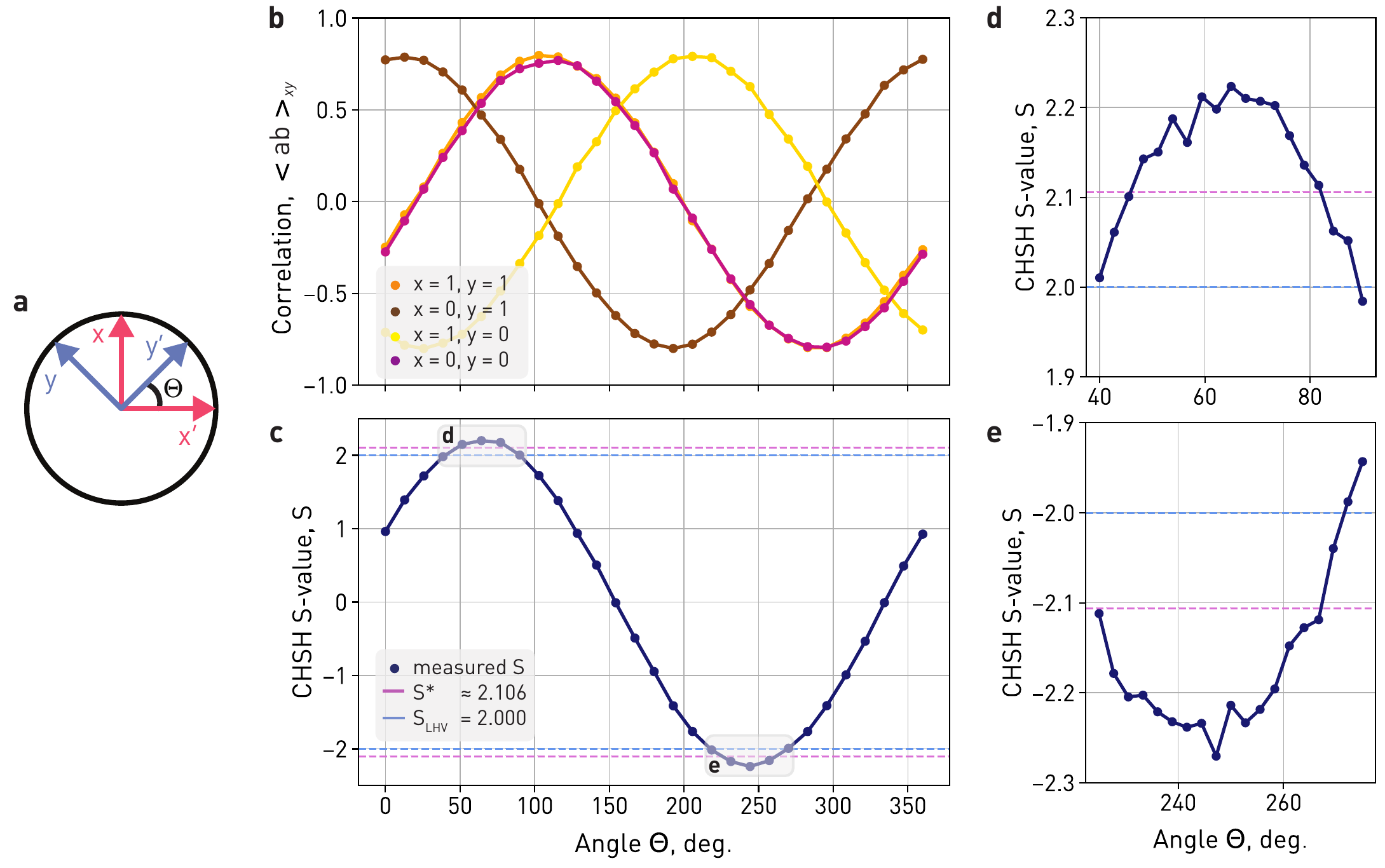}
  \caption{\textbf{Measured Bell inequality violation at different basis settings.} CHSH S-value plotted against the measurement basis offset angle $\Theta$. a, Illustration of the measurement bases (x$\perp$x',y$\perp$y') at the two nodes A and B, offset from each other by the measurement basis offset angle $\Theta$. b, Correlators $\langle a \cdot b\rangle_{(x,y)}$ of 29 individual Bell test runs over the full range of $\Theta$. Each point represents a Bell test comprising of n=36'157 trials. The lines connecting the data are guides to the eye. c, Corresponding CHSH S-value for each individual Bell test at offset angle $\Theta$, calculated from the correlators in panel b. The dashed lines correspond to the local hidden variable bound $S_\mathrm{LHV}=2$ (blue) and the self-testing bound $S^*\approx2.106$ (purple). d and e, CHSH S-value of Bell tests for a more narrow selection of $\Theta$ in the optimal angle regions indicated in panel c, comprising of n=55'188 trials each. In the self-testing experiment presented in the main text, we chose a measurement basis offset angle $\theta=245$~deg.}
  \label{fig:bell_angle}
  \end{center}
\end{figure*}

\paragraph{Bell test}
The two nodes A and B have access to two measurement bases $x=0,~x'=1$ and $y=0,~y'=1$, correspondingly. In the physical implementation, these bases correspond to rotations of the state of the two qubits around specific axes in the Bloch sphere. We call the relative angle $\Theta$ between the two local measurement bases at the two nodes (Fig.~\ref{fig:bell_angle}a) the measurement basis offset angle. In an idealized implementation, where all devices responsible for controlling the state of the qubits are perfectly synchronized to the same reference frame and the two qubits are operated at the same frequencies, the highest Bell inequality violation is found at $\Theta=\pi/4$ and $\Theta=5\pi/4$. 
We perform a set of 29 (loophole-free) Bell tests, each comprising $n=36'157$ trials, over the full range of $\Theta$, to find the two optimal offset angle regions, see Figs.~\ref{fig:bell_angle}b and c. As expected, we observe that the optimal offset angle regions are separated by 180 degrees (or $\pi$, accordingly). We then perform a set of 19 Bell test measurements, each consisting of $n=55'188$ trials, over a more narrow range of $\Theta$ around the optimal offset angle points, see Figs.~\ref{fig:bell_angle}d and e. Based on these calibrations, we fix the measurement basis offset angle to $\Theta=245$~degrees for the main self-testing experiment presented in the article.

\section{Robust Self-Testing Analysis}
\label{app:theory_main}

\subsection{Self-testing in the ideal case}
In the scenario of interest for this work, the aim is to characterize a state and measurements solely from the CHSH value $S$. Although in the case $S=2\sqrt{2}$ we do expect that well defined qubit measurements are performed on a maximally entangled two-qubit state, we do not presume that the Hilbert space dimension and the measurement calibration are known and remain unchanged throughout the experiment. Self-testing defines such a device-independent characterization.

Let us delve more deeply into self-testing statements within the ideal case where $S=2\sqrt{2}$. Considering the characterization of a bipartite state of an unknown dimension $\rho_{AB} \in L(\mathcal{H}_A \otimes \mathcal{H}_B)$, self-testing guarantees the existence of local extraction (completely positive and trace preserving) maps 
\begin{eqnarray}
\Lambda_A &:& L(\mathcal{H}_A) \rightarrow L(\mathbb{C}^2)\\
\Lambda_B &:& L(\mathcal{H}_B) \rightarrow L(\mathbb{C}^2)
\end{eqnarray}
which identify a sub-system of local dimension 2 in a maximally entangled two qubit state $|\phi^+\rangle~=~\frac{1}{\sqrt{2}}(|00\rangle~+~|11\rangle)$, that is
\begin{equation}
(\Lambda_A \otimes \Lambda_B)[\rho_{AB}] = \ketbra{\phi^+}.
\end{equation}
Concerning the characterization of the measurement, Alice's apparatus $\mathcal{A}$ is described as a device taking as inputs a state $\rho_A \in L(\mathcal{H}_A)$ and the measurement setting (a classical value $x=\{0,1\}$ in a register labeled by $in$) and giving as output the result of the measurement (a~classical value $a=\{0,1\}$ in a register labeled by $out$), see Fig.~\ref{fig:fig_channel}. In the ideal setting, self-testing guarantees the existence of an injection map
\begin{equation}
V_A : L(\mathbb{C}^2) \rightarrow L(\mathcal{H}_A)
\end{equation}
such that $ \forall \sigma \in L(\mathbb{C}^2)$ and $\forall x \in \{0,1\}$, 
\begin{equation}
\label{def_selftesting_ideal}
(\mathcal{A} \circ V_A)[\sigma \otimes |x\rangle\langle x|_{in}] = \tilde{\mathcal{A}}[\sigma \otimes |x\rangle\langle x|_{in}].
\end{equation}
The reference apparatus $\tilde{\mathcal{A}}$ corresponds to a pair of orthogonal Pauli $\sigma_z=|0\rangle\langle0|-|1\rangle\langle 1|$ and $\sigma_x=|+\rangle\langle +|-|-\rangle\langle -|$ measurements for the setting choice $x=0$ and $x=1$ respectively, that is
\begin{equation}
\begin{split}
\label{ideal_instrument}
        \tilde{\mathcal{A}}[\sigma \otimes |x\rangle\langle x|_{\text{in}}]= &  \bra{0} \sigma \ket{0} \langle x | 0\rangle_{\text{in}} \ketbra{0}_{\text{out}} \\
        & + \bra{1} \sigma \ket{1} \langle x | 0\rangle_{\text{in}} \ketbra{1}_{\text{out}}\\
        & + \bra{+} \sigma \ket{+} \langle x | 1\rangle_{\text{in}} \ketbra{0}_{\text{out}} \\
        & + \bra{-} \sigma \ket{-} \langle x | 1\rangle_{\text{in}} \ketbra{1}_{\text{out}}.
\end{split}
\end{equation}
In a scenario where Alice shares an entangled state and copies of the input register with a third party, appropriate projections performed by the third party cover all possible inputs for Alice. The equality in Eq.~\eqref{def_selftesting_ideal} can thus be expressed with a single bi-partite quantum state and two input registers of the form $\rho^+~=~\ketbra{\phi^+} \otimes (\ketbra{00}+\ketbra{11})/2$, and self-testing guarantees the existence of the map $V_A$ such that
 \begin{equation}
(\mathcal{A} \circ V_A)[\rho^+] = (\tilde{\mathcal{A}})[\rho^+]
\end{equation}
where $\mathcal{A} \circ V_A$ and $\tilde{\mathcal{A}}$ operate on Alice's inputs only.

We discuss below the device-independent state and measurement characterizations separately in the realistic case where $S < 2\sqrt{2}$. 

\subsection{State self-testing in the noisy case}

The case $S=2\sqrt{2}$ that we considered so far cannot be realized in an actual experiment due to unavoidable imperfections. Following Ref.~\cite{Kaniewski2016}, we characterize the actual state $\rho_{AB}$ in a realistic test with $S<2\sqrt{2}$  by the best fidelity one can extract with respect to the ideal state $|\phi^+\rangle$  
\begin{equation}
    \mathcal{F} (\rho_{AB}, \ketbra{\phi^+}) = \max_{\Lambda_A, \Lambda_B} F((\Lambda_A \otimes \Lambda_B) \rho_{AB}, \ketbra{\phi^+}).
\end{equation}
F is the square of the Ulhmann fidelity: for two density matrices $\tau$ and $\tau'$, $F(\tau,\tau')~=~\text{Tr}(\sqrt{\tau^{1/2}\tau'\tau^{1/2}})^2)$
which here reduces to the overlap ($F (\tau,\tau') = \text{Tr} [\tau \tau' ]$) since one of the two states is pure. 

To bound $\mathcal{F} (\rho_{AB},\ketbra{\phi^+})$, we take into account the fact that the state was used in an experiment leading to a given CHSH value. Interestingly, its infimum over all states compatible with a CHSH value $S$, labelled $\mathcal{F}_\mathsf{s}(S),$ can be lower bounded by \begin{equation}\label{eq:Kaniewski}
    \mathcal{F}_\mathsf{s}(S) \geq \frac{1}{2}+\frac{1}{2} \cdot \frac{S-S^*}{2 \sqrt{2}-S^*}
\end{equation}
where $S^* = \frac{16+14\sqrt{2}}
{17}\approx 2.106$~\cite{Kaniewski2016}. This tells us that for any state used to obtain a CHSH value of at least $S$, we can extract a bipartite state with local maps whose fidelity to the singlet state is higher than the right hand side of Ineq.~\eqref{eq:Kaniewski}. This lower bound is plotted in Fig.~\ref{fig:S_vs_fid}. 

It is not currently known whether this bound is tight. However, it has been shown that a non-trivial self-testing statement cannot be obtained for any $S^*$ below $\approx 2.05$~\cite{Valcarce2020}. Other self-testing bounds can eventually yield a higher extractability than given by Eq.~\eqref{eq:Kaniewski} by leveraging more information than the mere CHSH value~\cite{Bancal2015,Valcarce2022}.

\subsection{Measurement self-testing in the noisy case}
In the spirit of what has been done for state self-testing, we characterize the actual measurement apparatus in the noisy case $S<2\sqrt{2}$ by the best fidelity one can get between $\mathcal{A} \circ V_A[\rho^+]$ and $\tilde{\mathcal{A}}[\rho^+]$ when maximizing the injection map, that is
\begin{equation} \label{eq:defmeasfidelity}
    \mathcal{F}(\mathcal{A},\tilde{\mathcal{A}}) = \max_{V_A} F((\mathcal{A} \circ V_A)[\rho^+], \tilde{\rho}),
\end{equation}
where $\tilde{\rho}:=\tilde{\mathcal{A}}[\rho^+]$. The infimum of $\mathcal{F}(\mathcal{A},\tilde{\mathcal{A}})$ over all instruments compatible with a CHSH value $S$ is labelled $\mathcal{F}_{\mathsf{m}}(S)$. Its value is given in the following proposition and plotted in Fig.~\ref{fig:S_vs_fid}.

\begin{figure}
  \begin{center} \includegraphics[width=1\columnwidth]{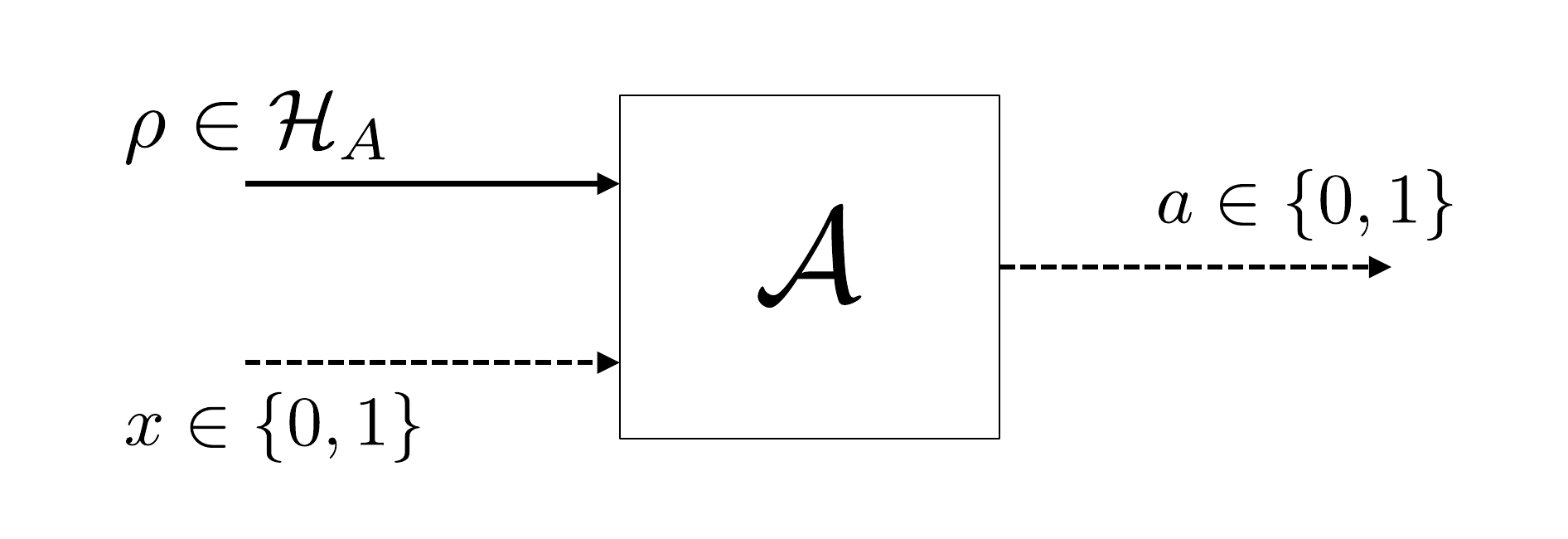}
   \caption{\textbf{Schematic description of a measurement apparatus $\mathcal{A}$ of one party in a Bell-experiment with two inputs and two outputs.} Formally speaking, it is a quantum channel with two separated inputs: a classical bit $x$, encoding the measurement choice, and a quantum state $\rho$, being measured, in an Hilbert space $\mathcal{H}_A$ of unknown dimension; and one output, consisting of a classical bit $a$ encoding the measurement outcome.}
  \label{fig:fig_channel}
  \end{center}
\end{figure}

\begin{figure}
  \begin{center} \includegraphics[width=1\columnwidth]{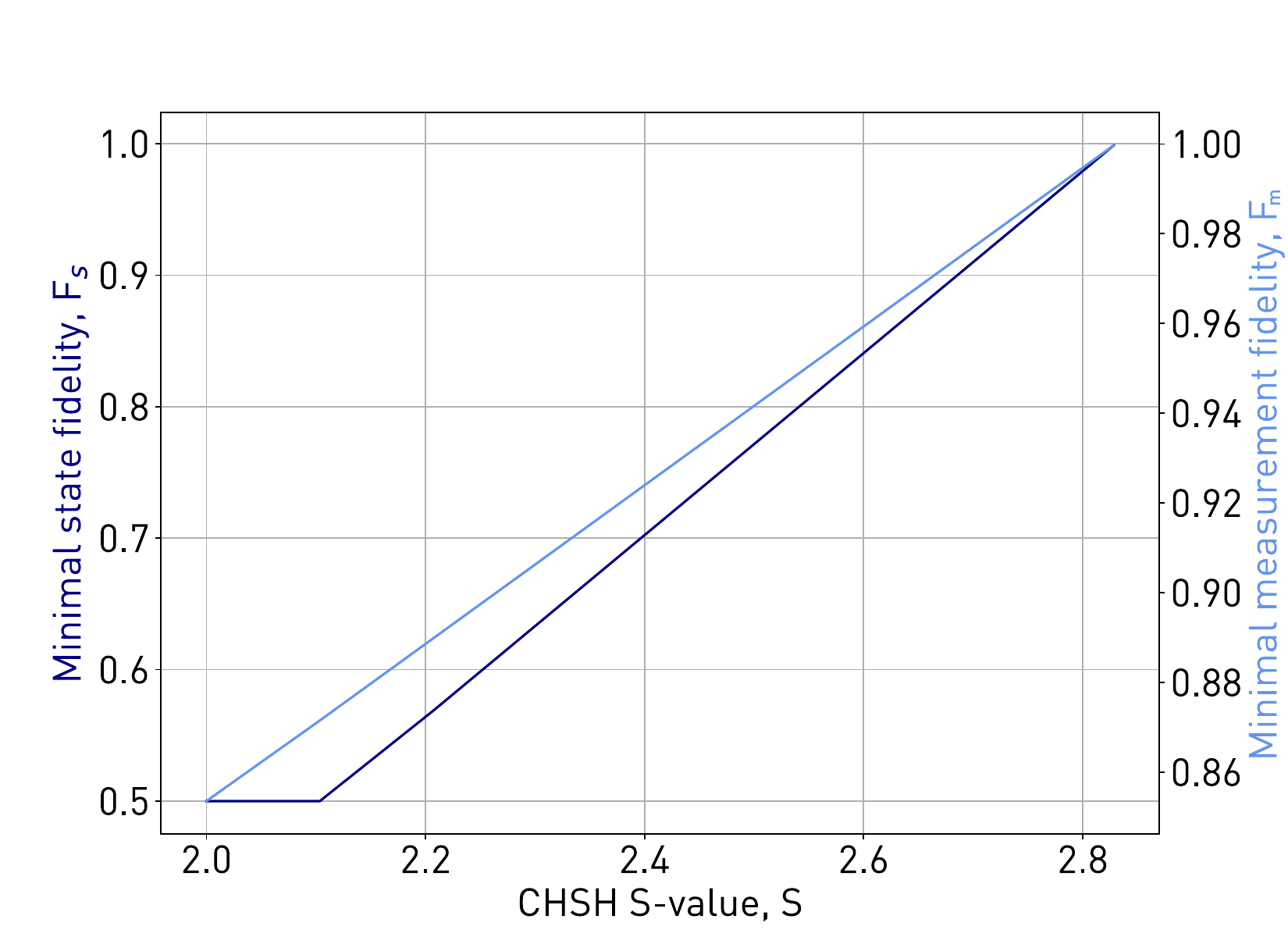}
  \caption{\textbf{Relation of the CHSH S-value with the self-testing results.} The figure shows the minimal Bell state fidelity (dark blue) and measurement fidelity (light blue) certifiable with a given measured CHSH S-value, see Eqs.~1 and ~2 in the main text.}
  \label{fig:S_vs_fid}
  \end{center}
\end{figure}

\begin{Proposition}
The minimal measurement fidelity with respect to the CHSH value $S$ is equal to:
\begin{equation}
    \mathcal{F}_{\mathsf{m}}(S) = \frac{ \sqrt{2} \, S +4}{8}.
\end{equation}
\end{Proposition}

This tells us that for any measurement apparatus $\mathcal{A}$, of arbitrary dimension, used to obtain a CHSH value at least $S$, there exists a local injection map such that it can be used as the ideal pair of Pauli measurements $\tilde{\mathcal{A}}$ with a fidelity at least equal to $\mathcal{F}_\mathsf{m}(S)$. 

The trivial bound for the fidelity is given for $S=2$, that is $\mathcal{F}_\mathsf{m}(2) = \frac{2\sqrt{2}+4}{8} \approx 0.854$. This corresponds to the fidelity one can always obtain with any apparatus acting on a qubit space by choosing the proper injection. This value of fidelity is attainable by an apparatus that can only give rise to local CHSH realizations ($S~\leq~2$). This is also the best fidelity one can obtain with an apparatus realizing measurements in only one basis (discarding the classical input bit). The following subsections are dedicated to the proof of this result, which is the main theoretical contribution of this publication.

\subsubsection{General case}

We first consider the general case of a Bell scenario.

\begin{Proposition}
The optimization on all possible quantum apparatus $\mathcal{A}$ can be restricted to an optimization on an apparatus with a qubit input space only.
\end{Proposition}

\begin{proof}[Proof of Proposition 2] In a Bell scenario with two binary measurements on each party, one can find local bases in which Alice's and Bob's measurements are block diagonal with blocks of size 2 according to Jordan's lemma. The full state shared by Alice and Bob can then be written as a convex mixture of the state on each block. The total CHSH value is the convex mixture of the CHSH value on each qubit-qubit block, that is $S = \sum_{ij} \lambda_{ij} S^{ij}$, with $\sum_{ij} \lambda_{ij} = 1$. In particular, there exists at least one block such that $S^{ij} \geq S$, and one can consider, up to relabeling, that this happens on Alice's first block. Let's denote $\mathcal{A}^0$ the apparatus of Alice restricted to this subspace.

Let's introduce the canonical injection map $\id_{\mathbb{C}^2\to \mathcal{H}_A^0}$ where $\mathcal{H}_A^0$ denotes the Hilbert space associated to the first block of Alice's input space. Suppose that there exists an injection map $V_A^0$ such that $\mathcal{F}(\mathcal{A}^0\circ V_A^0, \tilde{\mathcal{A}}) \geq \mathcal{F}^\star(S)$. Then, the injection map $ V_A := V_A^0 \circ \id_{\mathbb{C}^2\to \mathcal{H}_A^0}$  grants a fidelity $\mathcal{F}(\mathcal{A}\circ V_A, \tilde{\mathcal{A}})=\mathcal{F}(\mathcal{A}^0\circ V_A^0, \tilde{\mathcal{A}}) \geq \mathcal{F}^\star(S)$. Thus the function $\mathcal{F}^\star$ is still a lower bound on the minimal fidelity $\mathcal{F}_\mathsf{m}$ regardless of the Hilbert space dimensions. 

This argument allows us to reduce the optimization onto qubit apparatuses, as any lower bound of those would still be valid overall. Note that if the lower bound on qubit measurement apparatuses is tight, i.e.~if there is an explicit instrument achieving this bound, then the overall bound is also tight.
\end{proof}

\subsubsection{Qubit case}

If the Hilbert space dimension of the input of Alice's apparatus is 2, and up to a local basis rotation which can be absorbed in the injection $V_A$, then $\mathcal{A}$ can be written as
\begin{equation} \label{eq:qubitchannel}
\begin{split}
        \mathcal{A}(\alpha)[\sigma \otimes |x\rangle\langle x|_{\text{in}}]= &  \bra{0} \sigma \ket{0} \langle x | 0\rangle_{\text{in}} \ketbra{0}_{\text{out}} \\
        & + \bra{1} \sigma \ket{1} \langle x | 0\rangle_{\text{in}} \ketbra{1}_{\text{out}}\\
        & + \bra{+_\alpha} \sigma \ket{+_\alpha} \langle x | 1\rangle_{\text{in}} \ketbra{0}_{\text{out}} \\
        & + \bra{-_\alpha} \sigma \ket{-_\alpha} \langle x | 1\rangle_{\text{in}} \ketbra{1}_{\text{out}}
\end{split}
\end{equation}
where $\alpha\in [0,\pi]$, $\ket{+_\alpha}= \cos(\alpha/2)\ket{0}+ \sin(\alpha/2)\ket{1}$ and $\ket{-_\alpha}= \sin(\alpha/2)\ket{0}- \cos(\alpha/2)\ket{1}$.
\begin{Proposition}
    The fidelity of the measurement apparatus $A(\alpha)$ defined in~Eq.\eqref{eq:qubitchannel} with respect to the ideal one given in Eq.~\eqref{ideal_instrument} is given by
    \begin{equation}
    \mathcal{F}(\alpha) = \frac{1}{4}\left(2+\sqrt{2}\cos\left(\frac{\alpha}{2}\right) + \sqrt{2}\sin\left(\frac{\alpha}{2}\right)\right).
    \end{equation}
\end{Proposition}

\begin{proof}[Proof of Proposition 3]
The measurement fidelity between the two apparatuses is given by
\begin{equation}
     \mathcal{F}(\mathcal{A}(\alpha),\tilde{\mathcal{A}}) = \max_{V_A} F((\mathcal{A}(\alpha)\circ V_A)[\rho^+], \tilde{\rho})
\end{equation}
where in the qubit case, the injection map $V_A$ is a CPTP map between two Hilbert spaces of dimension two. 
Introducing $C_{V_A}=2V_A\left[\ketbra{\phi^+}\right]$, the Choï matrix of $V_A$, we have
\begin{equation}
\begin{split}
     &F((\mathcal{A}(\alpha)\circ V_A)[\rho^+], \tilde{\rho}) \\
     & = F(\mathcal{A}(\alpha)\left[V_A[\ketbra{\phi^+}]\otimes \frac{1}{2}(\ketbra{00}+\ketbra{11})\right], \tilde{\rho}) \\
     & = \frac{1}{16}\huge{\text{(}} \sqrt{\bra{00}C_{V_A}\ket{00}} + \sqrt{\bra{11}C_{V_A}\ket{11}}\\
     & \quad + \sqrt{\bra{+_\alpha +}C_{V_A}\ket{+_\alpha +}} + \sqrt{\bra{-_\alpha -}C_{V_A}\ket{-_\alpha -}} \huge{\text{)}}^2.\\
\end{split}
\end{equation}
Using the Cauchy-Schwartz inequality $(\sum_i \sqrt{x_i})^2 \leq n \sum_i x_i$, one can then upper bound the previous quantity by
\begin{equation}
\label{Fid}
\begin{split}
    F((\mathcal{A}&(\alpha)\circ V_A)[\rho^+], \tilde{\rho}) \\
    & \leq \frac{1}{4}\huge{\text{(}} \bra{00}C_{V_A}\ket{00} + \bra{11}C_{V_A}\ket{11}\\
    & \quad + \bra{+_\alpha +}C_{V_A}\ket{+_\alpha +} + \bra{-_\alpha -}C_{V_A}\ket{-_\alpha -} \huge{\text{)}}\\
    & \leq \frac{1}{4}\text{Tr}(M(\alpha)\cdot C_{V_A})
\end{split}
\end{equation}
where 
\begin{equation}
    M(\alpha) = \ketbra{00} + \ketbra{11} + \ketbra{+_{\alpha}+} + \ketbra{-_{\alpha}-}.
\end{equation}
Note that this upper bound on $F((\mathcal{A}(\alpha)\circ V_A)[\rho^+], \tilde{\rho})$ is tight whenever the four values $\bra{00}C_{V_A}\ket{00}$, $\bra{11}C_{V_A}\ket{11}$, $\bra{+_\alpha+}C_{V_A}\ket{+_\alpha+}$ and $\bra{-_\alpha-}C_{V_A}\ket{-_\alpha-}$ are equal. 

By taking the maximum over all possible injection maps $V_A$, we have 
\begin{equation}
    \mathcal{F}(\mathcal{A}(\alpha),\tilde{\mathcal{A}}) \leq \frac{1}{4} \max_{\substack{C_{V_A} \in L(\mathbb{C}^2\otimes \mathbb{C}^2) \\ \text{s.t.} \ \text{Tr}_{\mathbb{C}^2} (C_{V_A}) = \id_\mathbb{C}^2, \, C_{V_A} \succeq 0}}  \  \text{Tr}( M(\alpha)\cdot C_{V_A}) .
\end{equation}
The right-hand term is a semi-definite positive optimisation (SDP) the dual of which is given by
\begin{equation}
    \min_{\substack{L\in L(\mathbb{R}^2)\\ \text{s.t.} \ \id \otimes L \succeq M(\alpha)}}  \ \text{Tr}(L). \\
\end{equation}
Due to weak duality, any solution to this dual SDP is an upper bound on the primal SDP. One solution is given by
\begin{equation}
\label{Choice_L}
    L = \frac{1}{2} \left(2+\sqrt{2(1+\sin(\alpha)}\right) \id.
\end{equation}
Now we have $M(\alpha)-\id \otimes L = N(\alpha) - \sqrt{2(1+\sin(\alpha)} \, \id$ where $N(\alpha)$ is given by
\begin{equation}
    N(\alpha) = \begin{pmatrix}
        1 & \cos(\alpha) & 0 &  \sin(\alpha) \\
        \cos(\alpha) & -1 & \sin(\alpha) & 0 \\
        0 & \sin(\alpha) & -1 & -\cos(\alpha) \\
        \sin(\alpha) & 0 & -\cos(\alpha) & 1
    \end{pmatrix}.
\end{equation}
The eigenvalues of $M(\alpha)-\id \otimes L$ are thus the eigenvalues of $N(\alpha)$ translated by $-\sqrt{2(1+\sin(\alpha)}$. To show that $M(\alpha) -\id \otimes L \preceq 0$, it is sufficient to show the largest eigenvalue of $N(\alpha)$ is smaller than $\sqrt{2(1+\sin(\alpha)}$. This is the case as the characteristic polynomial of $N(\alpha)$ is given by 
\begin{equation}
    |N(\alpha)-X\id| = X^4 - 4X^2 + 2(1+\cos(2\alpha)).
\end{equation}
The roots of this polynomial are given by $X = \pm \sqrt{2(1 \pm \sin(\alpha))}$ and thus the largest eigenvalue of $N(\alpha)$ is indeed $\sqrt{2(1+\sin(\alpha)}$.\\
The value of the dual SDP for the choice of $L$ given in Eq.~\eqref{Choice_L} is $2+\sqrt{2(1+\sin(\alpha)}$ which upper bounds the primal SDP problem. Thus:
\begin{equation}
\begin{split}
    \mathcal{F}(\mathcal{A}(\alpha),\tilde{\mathcal{A}}) \leq & \frac{1}{4} (2+\sqrt{2(1+\sin(\alpha)}) \\
    & = \frac{1}{4}\left(2+\sqrt{2}\cos\left(\frac{\alpha}{2}\right) + \sqrt{2}\sin\left(\frac{\alpha}{2}\right)\right)
\end{split}.
\end{equation}
Note that this upper bound is actually tight as it can be reached with a unitary rotation $V_A = \begin{pmatrix}
    \cos(\theta) & \sin(\theta) \\
    - \sin(\theta) & \cos(\theta) 
\end{pmatrix}$ of angle $\theta = - \alpha/4 + \pi/8$ which verifies the equality condition of the Cauchy-Schwartz inequality, see~\eqref{Fid}.
\end{proof}

We now need to include the fact that $\mathcal{A}$ was used to obtain a CHSH value $S$. This is made possible by bounding the range of parameter $\alpha$ with respect to the $S$-value. 
\begin{Proposition}
    A Bell experiment involving the measurement apparatus $\mathcal{A}(\alpha)$ on Alice's side can only reach CHSH values up to $2\sqrt{2}\sin(\frac{\alpha}{2}+\frac{\pi}{4})$. \\
    Conversely, when reaching a CHSH value $S$ and Alice performs qubit measurements, the value of the measurement parameter $\alpha$ of Alice's apparatus belongs to
    \begin{equation}
    \alpha \in \left[2 \arcsin \left(\frac{S}{2\sqrt{2}}\right) - \frac{\pi}{2} , \frac{3\pi}{2} - 2 \arcsin \left(\frac{S}{2\sqrt{2}}\right)\right].
\end{equation}
\end{Proposition}

\begin{proof}[Proof of Proposition 4] The realization of the CHSH scenario with qubits can be parametrized in the following way (up to local unitaries): The shared state is some density matrix $\rho_{AB}$, Alice performs $\mathcal{A}(\alpha)$ and Bob $\mathcal{B}(\beta)$, with $\alpha,\beta \in [0,\pi)$. The CHSH value is given by 
\begin{equation}
    S(\alpha,\beta,\rho_{AB}) = \text{Tr}(M_{\alpha,\beta} \, \rho_{AB}),
\end{equation}
where the operator $M_{\alpha,\beta}$ is 
\begin{equation}
\begin{split}
     M_{\alpha,\beta} &= \sigma_z \otimes \sigma_z + \sigma_z \otimes (c_\beta \sigma_z + s_\beta \sigma_x)\\
    & + (c_\alpha \sigma_z + s_\alpha \sigma_x) \otimes \sigma_z  \\
    &- (c_\alpha \sigma_z + s_\alpha \sigma_x) \otimes (c_\beta \sigma_z + s_\beta \sigma_x).
\end{split}
\end{equation}
In the previous equation, we defined $c_\alpha=\cos(\alpha)$, $s_\alpha=\sin(\alpha)$ and similarly for $c_\beta$ and $s_\beta$.

Since $M_{\alpha,\beta}$ is a real symmetric matrix, the best value of $S(\alpha,\beta,\rho_{AB})$ for a given $\alpha,\beta$ is the largest eigenvalue of $M_{\alpha,\beta}$. First of all, note that
\begin{equation}
    M_{\alpha,\beta}^2 = 4 (1+s_\alpha s_\beta \, \sigma_y \otimes \sigma_y).
\end{equation}
Now since $\sigma_y \preceq \id$, we have $0 \preceq M_{\alpha,\beta}^2 \preceq 4(1+s_\alpha s_\beta)$ and any of the eigenvalues $\lambda$ of $M_{\alpha,\beta}^2$ is smaller than this upper bound. Since $M_{\alpha,\beta}$ can be diagonalized, its eigenvalues are of the form $\pm \sqrt{\lambda}$. Its largest eigenvalue is thus smaller than $2 \sqrt{1+\sin(\alpha)\sin(\beta)} \leq 2\sqrt{1+\sin(\alpha)}$. Finally, the largest possible CHSH value that one can obtain for a given angle $\alpha$ is upper bounded by 
\begin{equation}
    S(\alpha,\beta,\rho_{AB}) \leq 2\sqrt{1+\sin(\alpha)} = 2\sqrt{2}\sin\left(\frac{\alpha}{2}+\frac{\pi}{4}\right).
    \label{CHSH_bound}
\end{equation}
Note that this upper bound is tight as it can be obtained with the state $\rho_{AB}=\ketbra{\phi^+}$, measurement apparatus $\mathcal{A}(\alpha)$ for Alice and Bob performing measurements with $B_0 = \cos(\alpha/2)\sigma_z + \sin(\alpha/2) \sigma_x$ and $B_1 = \sin(\alpha/2)\sigma_z + \cos(\alpha/2) \sigma_x$. 

This upper bound implies that if $S(\alpha,\beta,\rho_{AB}) \geq S$, we have $\sin(\alpha/2 + \pi/4) \geq \frac{S}{2\sqrt{2}}$. Hence, the parameter $\alpha$ can only take values in
\begin{equation}
    \alpha \in \left[2 \arcsin \left(\frac{S}{2\sqrt{2}}\right) - \frac{\pi}{2} , \frac{3\pi}{2} - 2 \arcsin \left(\frac{S}{2\sqrt{2}}\right)\right].
\end{equation}
\end{proof}

\begin{proof}[Proof of Proposition 1] Proposition 2 allows one to reduce to the case where the $\mathcal{A}$ input state is a qubit. In this case, it can be described by a single parameter $\alpha$ as in Eq.~(\ref{eq:qubitchannel}). The range of $\alpha$ can be restricted with the $S$-value using Proposition 4. Since the maximal fidelity of any $\mathcal{A}(\alpha)$ is given by $\mathcal{F}(\alpha)$ (see Proposition 3), and this function is symmetric under $\alpha \to \pi -\alpha$ and increasing on $[0,\pi/2]$, the minimal measurement fidelity value for a CHSH value $S$ is attained at the lowest possible $\alpha$ in $[0,\pi/2]$. This can be computed explicitly and leads to 
\begin{equation}
    \mathcal{F}_\mathsf{m} (S) = \frac{ \sqrt{2} \, S +4}{8}.
    \label{Uhl_bound}
\end{equation}
\end{proof}

\section{Tomographic Fidelity of Measurements}
\label{app:measurement}

The fidelity of measurements introduced in the previous section in the self-testing scenario can be adapted in the tomographic context. In that case, two main assumptions are made: the dimension of the Hilbert spaces are trusted and the choice of the local basis is fixed. Those translate to operational assumptions of the measurement fidelity: the physical map $\mathcal{A}$ takes a qubit $\sigma \in \mathbb{C}^2$ as input and the injection map is set to the identity ($V_A = \id_{\mathcal{H}_A}$). The action of the physical measurement apparatus can thus be described as:
\begin{equation}
    \mathcal{A}[\sigma \otimes \ketbra{x}_{\text{in}}] = \sum_{a=0,1} \text{Tr}(M_{a,x}\sigma) \otimes \ketbra{a}_{\text{out}}
\end{equation}
where for each input $x$ the operators $\{M_{a,x}\}_k$ encode a POVM measurement. In the ideal case (uniquely reaching $S=2\sqrt{2}$), those POVMs are projective measurements $\ketbra{u_{a,x}}$ on states $\ket{0}$ and $\ket{1}$ for input $x=0$, and states $\ket{+}$ and $\ket{-}$ for input $x=1$. The measurement fidelity is given by:
\scriptsize
\begin{equation}
    \begin{split}
        &\mathcal{F} (\mathcal{A}, \tilde{\mathcal{A}}) = F(\mathcal{A}[\rho^+], \tilde{\mathcal{A}}[\rho^+])\\
        & = \left(\text{Tr} \sqrt{\sqrt{\tilde{\mathcal{A}}[\rho^+]}\mathcal{A}[\rho^+]\sqrt{\tilde{\mathcal{A}}[\rho^+]}}\right)^2 \\
        & = \frac{1}{4}\left(\text{Tr} \sqrt{\tilde{\mathcal{A}}[\rho^+]\mathcal{A}[\rho^+]\tilde{\mathcal{A}}[\rho^+]}\right)^2 \\
        & = \frac{1}{16}\left(\text{Tr} \sqrt{\sum_{a,x=0,1} \bra{u_{a,x}} M_{a,x} \ket{u_{a,x}} \ketbra{u_{a,x}} \otimes \ketbra{ax} }\right)^2 \\
        & = \frac{1}{16}\left(\text{Tr} \sum_{a,x=0,1} \sqrt{\bra{u_{a,x}} M_{a,x} \ket{u_{a,x}}}\ketbra{u_{a,x}} \otimes \ketbra{ax} \right)^2 \\
        & = \frac{1}{16}\left(\sum_{a,x=0,1} \sqrt{\bra{u_{a,x}} M_{a,x}\ket{u_{a,x}}} \right)^2
    \end{split}
\end{equation}
\normalsize
Note that all quantities in the above expression can be directly related to probabilities of measurements accessible in the laboratory by state preparation:
\begin{equation}
\begin{split}
    \mathbb{P}_\mathsf{c}(a|\ket{u_{a,x}}, x) & = \bra{a} \mathcal{A}[\ketbra{u_{a,x}}\otimes \ketbra{x}] \ket{a} \\
    & = \bra{u_{a,x}} M_{a,x}\ket{u_{a,x}}
\end{split}
\end{equation}
This probability corresponds to obtaining the expected outcome of the ideal apparatus, knowing the input state. Using the above formula:
\begin{equation} \label{eq:tomographic_fid}
    \mathcal{F} (\mathcal{A}, \tilde{\mathcal{A}}) = \frac{1}{16}\left(\sum_{a,x=0,1} \sqrt{\mathbb{P}_\mathsf{c}(a|\ket{u_{a,x}}, x}) \right)^2.
\end{equation}

While the above quantity can be accessible in theory, we only have access in practice to the probabilities $\mathbb{P}_\mathsf{r}$ and $\mathbb{P}_\mathsf{z}$ corresponding respectively to the $\pi/2$-rotation on the Bloch sphere and the measurement in the $\sigma_z$-basis. Those quantities can be directly related to the fidelity of some physical operation with respect to an ideal one in the tomographic context. As such 
\begin{subequations}
    \begin{equation}
        \mathcal{F}(\mathcal{R}, \tilde{\mathcal{R}_\frac{\pi}{2}}) = \frac{1}{4} \left( \sum_{x=0,1} \sqrt{\mathbb{P}_\mathsf{r}(\ket{u_{1,x}}| \ket{x})} \right)^2
    \end{equation}
    \begin{equation}
        \mathcal{F}(\mathcal{M}, \tilde{\mathcal{M}_z}) = \frac{1}{4} \left( \sum_{a=0,1} \sqrt{\mathbb{P}_\mathsf{z}(a| \ket{a})} \right)^2
    \end{equation}
\end{subequations}
where the operators $\mathcal{R}$ and $\mathcal{M}$ describe the physical (imperfect) implementation of the basis rotation and the $\sigma_z$ measurement. For example, when considering that the readout error is symmetric on all possible states prepared, i.e. $\mathbb{P}_\mathsf{z}(a'| \ket{a}) = \mathbb{P}_\mathsf{z}(a| \ket{a'})$, we obtain that the $\sigma_z$-fidelity is $\mathcal{F}(\mathcal{M}, \tilde{\mathcal{M}_z}) = \mathbb{P}_\mathsf{z}(0| \ket{0})$.

The overall physical apparatus $\mathcal{A}$ is obtained by performing the basis rotation only when the input is $x=1$ and in all cases performing a $\sigma_z$ measurement afterwards. The probabilities $\mathbb{P}_\mathsf{c}$ therefore directly relate to the probabilities $\mathbb{P}_\mathsf{r}$ and $\mathbb{P}_\mathsf{z}$. In particular, when the rotation is performed:
\begin{equation}
\begin{split}
    \mathbb{P}_\mathsf{c}(0|\ket{+}, 1) = \mathbb{P}&_\mathsf{r}(\ket{0}| \ket{+})\mathbb{P}_\mathsf{z}(0| \ket{0}) \\
    & + \mathbb{P}_\mathsf{r}(\ket{1}| \ket{+})\mathbb{P}_\mathsf{z}(0| \ket{1}) \\
    & - 2 (\bra{0} \mathcal{R}[\ketbra{+}] \ket{1})(\bra{0} \mathcal{M}[\bra{0}\ket{1}] \ket{0}) 
\end{split}
\end{equation}
The last term of the above equation can be bounded in absolute value using the CPTP property of channels $\mathcal{R}$ and $\mathcal{M}$:
\begin{subequations}
    \begin{equation}
        |\bra{0} \mathcal{R}[\ketbra{+}] \ket{1}|^2 \leq \mathbb{P}_\mathsf{r}(\ket{0}| \ket{+}) \mathbb{P}_\mathsf{r}(\ket{1}| \ket{+})
    \end{equation}
    \begin{equation}
        |\bra{0} \mathcal{M}[\ket{0}\ket{1}] \ket{0}|^2 \leq \mathbb{P}_\mathsf{z}(0| \ket{0}) \mathbb{P}_\mathsf{z}(1| \ket{0})
    \end{equation}
\end{subequations}
If one supposes that the measured physical errors, $\mathbb{P}_\mathsf{r}(\ket{1}| \ket{+}):= \varepsilon_\mathsf{r}$ and $ \mathbb{P}_\mathsf{z}(1| \ket{0}):=\varepsilon_\mathsf{z}$, are small, then the overall probability of having the correct output, when state $\ket{+}$ is prepared and input $x=1$ is given, can be lower bounded at first order by:
\begin{equation}
    \mathbb{P}_\mathsf{c}(0|\ket{+}, 1) \geq 1-\varepsilon_\mathsf{r}-\varepsilon_\mathsf{z} - 2\sqrt{\varepsilon_\mathsf{r}}\sqrt{\varepsilon_\mathsf{z}}
\end{equation}
If one makes the extra assumptions that all errors are symmetric and that performing no rotation leads to smaller errors, then this quantity lower bounds all other probabilities $\mathbb{P}_\mathsf{c}(a|\ket{u_{a,x}}, x)$. Taking this into account in Eq.~(\ref{eq:tomographic_fid}) grants
\begin{equation}
    \mathcal{F} (\mathcal{A}, \tilde{\mathcal{A}}) \geq \mathbb{P}_\mathsf{c}(0|\ket{+}, 1).
\end{equation}
Combining the two previous equations gives a lower bound on the overall tomographic measurement fidelity
\begin{equation}
    \mathcal{F} (\mathcal{A}, \tilde{\mathcal{A}}) \geq 1-\varepsilon_\mathsf{r}-\varepsilon_\mathsf{z} - 2\sqrt{\varepsilon_\mathsf{r}}\sqrt{\varepsilon_\mathsf{z}}.
\end{equation}

Note that in the case of symmetric errors, $\varepsilon_\mathsf{z}$ can be directly related to the readout fidelities presented in Appendix~\ref{app:calibrations}.a by the following formula:
\begin{equation}
    \varepsilon_\mathsf{z, A(B)} = \frac{1}{2}(1-F_{r,A(B)}) 
\end{equation}
To obtain a bound valid on both parties, one can always assume $\varepsilon_\mathsf{z} = \min\{\varepsilon_\mathsf{z, A},\varepsilon_\mathsf{z, B}\}$. As the measured error probabilities in the experimental setup are $\varepsilon_\mathsf{r} = 0.25\%$ and $\varepsilon_\mathsf{z}=1.4\%$, we obtain a lower-bound tomographic measurement fidelity $\mathcal{F}(\mathcal{A}, \tilde{\mathcal{A}}) \geq 97.2\%$.

\section{Timing Verification}
\label{app:timing}
In this experiment, we follow the characterization of the space-time configuration discussed in Ref.~\cite{Storz2023a}, relevant for ensuring that the locality loophole is closed through space-like separation of the start and stop events of each trial.
As mentioned in the main text, and similarly to Ref.~\cite{Storz2023}, we define the start time $t_\star$ of each trial of the self-testing protocol as the moment when the local random number generators start generating a random input bit. The stop time $t_\times$ marks the moment when the measurement signal arrives at the input of the analog-to-digital converter (ADC), which is part of the trusted laboratory. The locality loophole is closed when the following conditions hold: the basis choice at node A must not be able to influence the measurement outcome at node B using influences propagating at most at the speed of light in vacuum $c$, and vice versa. In order to ensure that, we must precisely know where in space and when in time these events happen.

We determine the spatial distances between the locations corresponding to the start event at one node and the stop event at the other node, as discussed in Ref.~\cite{Storz2023}. We measure the shorter of these two distances to be $d=32.928\,\rm{m}$, corresponding to an available time budget for each trial of $t_d=d/c=109.83$~ns, during which classical communication between the two nodes is forbidden by the laws of special relativity. The standard deviation on this value is on the order of ten picoseconds. We note that $d$ is a few centimeters larger than in Ref.~\cite{Storz2023} because of a minor rearrangement of the room-temperature electronics.

We further measure the actual duration of a trial in a scheme outlined in Ref.~\cite{Storz2023a}. We conclude that the duration of the protocol was no longer than $t_{\mathrm{protocol}}=106.7$~ns, and therefore $t_d-t_{\mathrm{protocol}}=3.1$~ns shorter than the time budget imposed by the space-like separation between the start and stop events at the two nodes. The standard deviation of the mean value is on the order of a few hundred picoseconds. We note that this margin is about 50$\%$ higher than for the loophole-free Bell test presented on the same setup \cite{Storz2023}, constituting about 2.8$\%$ of the available time budget. The margin is therefore in a similar regime as for related experiments \cite{Hensen2015,Giustina2015,Shalm2015,Li2018i}.

\bibliography{QudevRefDB,references}
\end{document}